\documentclass[10pt,journal,a4paper]{IEEEtran}
\usepackage{amssymb}
\usepackage{amsmath}
\usepackage{bm}
\usepackage{comment}
\usepackage{amsthm}
\usepackage{graphicx}
\usepackage{textcomp,gensymb}
\usepackage{subfigure}
\usepackage{flushend}
\usepackage{float}
\usepackage{cite}
\usepackage{enumerate}
\usepackage{enumitem}
\usepackage{multirow}
\usepackage{algorithm,algorithmicx}
\usepackage{algpseudocode}
\usepackage{amsfonts}
\usepackage{url}
\usepackage{diagbox}
\usepackage{color, soul}
\usepackage{amsthm}
\usepackage{type1cm}

\newcommand{\titlefont}{\fontsize{16.5pt}{24pt}\selectfont}

\newtheorem{remark}{\bf~~Remark}
\newtheorem{lemma}{\bf~~Lemma}
\newtheorem{proposition}{\bf~~Proposition}

\begin{document}
\title{{\titlefont Achievable Degrees of Freedom Analysis and Optimization in Massive MIMO via Characteristic Mode Analysis}}

\author{
	\IEEEauthorblockN{
		{Shaohua Yue}, \IEEEmembership{Graduate Student Member, IEEE},
		{Siyu Miao}, \IEEEmembership{Student Member, IEEE},
		{Shuhao Zeng}, \IEEEmembership{Member, IEEE},
		{Fenghan Lin}, \IEEEmembership{Senior Member, IEEE},
		{and Boya Di}, \IEEEmembership{Senior Member, IEEE}}
	
	\thanks{Shaohua Yue and Boya Di are with the State Key Laboratory of Photonics and Communications, School of Electronics, Peking University, Beijing 100871, China. (email: \{yueshaohua; boya.di\}@pku.edu.cn).} 
	\thanks{Shuhao Zeng is with Department of Electrical and Computer Engineering, Princeton University, NJ, USA. (email: shuhao.zeng96@gmail.com).}
	\thanks{Siyu Miao and Fenghan Lin are with the School of Information Science and Technology, ShanghaiTech University, Shanghai  201210, China. (email: \{miaosy;linfh\}@shanghaitech.edu.cn).} 
	\thanks{This work has been submitted to the IEEE for possible publication. Copyright may be transferred without notice, after which this version may no longer be accessible.}
}

\maketitle
\begin{abstract}
Massive multiple-input multiple-output (MIMO) is esteemed as a critical technology in 6G communications, providing large degrees of freedom (DoF) to improve multiplexing gain. This paper introduces characteristic mode analysis (CMA) to derive the achievable DoF. Unlike existing works primarily focusing on the DoF of the wireless channel,
the excitation and radiation properties of antennas are also involved in our DoF analysis, which influences the number of independent data streams for communication of a MIMO system. Specifically, we model the excitation and radiation properties of transceiver antennas using CMA to analyze the excitation and radiation properties of antennas. The CMA-based DoF analysis framework is established and the achievable DoF is derived. A characteristic mode optimization problem of antennas is then formulated to maximize the achievable DoF. A case study where the reconfigurable holographic surface (RHS) antennas are deployed at the transceiver is investigated, and a CMA-based genetic algorithm is later proposed to solve the above problem. By changing the characteristic modes electric field and surface current distribution of RHS, the achievable DoF is enhanced. 
Full-wave simulation verifies the theoretical analysis on the the achievable DoF and shows that, via the reconfiguration of RHS based on the proposed algorithm, the achievable DoF is improved.
\end{abstract}

\markboth{This work has been submitted to the IEEE for possible publication. Copyright may be transferred without notice, after which this version may no longer be accessible.}
{Shell \MakeLowercase{\textit{et al.}}: A Sample Article Using IEEEtran.cls for IEEE Journals}


\section{Introduction}\label{sec::intro}

Massive multiple-input multiple-output (MIMO) is a key enabling technology in wireless communication networks, which can fulfill the demands of spectral efficiency improvement for the future sixth generation (6G) communications\cite{6GMIMO,6GMIMO2,6GMIMO3}.
To achieve this, the physical aperture of antennas is scaled up, allowing the MIMO system to support a higher number of degrees of freedom (DoF). The DoF is typically defined as the number of eigenvalues of the MIMO channel matrix that exceed a certain threshold \cite{landaudof},
and it corresponds to the number of independent data streams available for effective data transmission\cite{dofstream}. As such, the DoF serves as a fundamental figure of merit for evaluating the spatial multiplexing gain of a MIMO system\cite{edofimportance}.

Existing research on DoF of MIMO systems can be mainly grouped into two categories based on different channels, i.e., the scattering channel\cite{dofwavenumber,DoFnear} and the line of sight (LoS) channel\cite{commode,multidof,af}. In the case of the scattering channel, it is modeled as the superposition of plane-wave modes of different angular directions, and then a Fourier plane-wave series expansion~\cite{wnd} is performed to analyze the DoF. Predicated on such an idea, the work\cite{dofwavenumber} studies the DoF of the isotropic scattering channel. 
The influence of the reactive near field on DoF of the scattering channel is studied in \cite{DoFnear}. 
Differently, the Green function-based channel modeling is used for DoF analysis for the LoS channel. 
In \cite{commode}, the DoF of the LoS channel between a large and a small intelligent surface is studied.
The DoF of a multiuser near-field communication scenario is discussed in \cite{multidof}, where the effect of spatial blocking between users' antennas on the DoF of the communication system is calculated. 
{In\cite{af}, the DoF of the near-field channel for five types of extremely-large scale MIMO designs are analyzed and compared.}

The above research on DoF analysis primarily focuses on the ideal case where only the wireless propagation environment is considered. However, the electromagnetic (EM) characteristics of antennas\footnote{The EM characteristics of antennas refer to properties that are related to the antenna excitation and radiation~\cite{antenna}, such as the antenna impedance~\cite{antennaimpedence}, antenna efficiency~\cite{mc}, and radiation pattern~\cite{dofantenna}.} also have a non-negligible influence on wireless signals between the transmitter and the receiver\cite{antennaimpedence}. 
Some initial works examine how dipole antennas \cite{mc} or simplified radiating objects \cite{dofcm} influence the DoF. In \cite{mc}, the effect of mutual coupling between antenna elements on the DoF is studied in the case of a one-dimension dipole antenna-based MIMO system. In \cite{dofcm}, the average shadow area of antennas are utilized to estimate the DoF of simplified radiating objects without considering the antenna ports. Our previous work \cite{dofantenna} investigates how the radiation pattern of antenna elements impacts the DoF of the isotropic scattering channel, but the excitation property of antennas is excluded. Note that wireless channels and EM characteristics of antennas are coupled. It still remains to be explored how to analyze and optimize the DoF while jointly considering the wireless channel and the excitation and radiation properties of various antennas.

In this paper, we aim to establish a framework to analyze the DoF where both the influence of antenna excitation and radiation characteristics as well as the wireless channel are involved. Two challenges have thus arisen. \emph{First}, it is non-trivial to model the excitation and radiation properties of antennas to ensure compatibility with communication models. Directly solving Maxwell's equations 
fails to derive closed-form expressions that capture the complete EM behavior for use in communication models. \emph{Second}, it is complicated to analyze the DoF when considering the EM properties of antennas. Such an analysis is based on the EM relationship between the input and output signals at the antenna ports. For this relationship, the EM properties of antennas are coupled with the wireless propagation channel, which alters the spatial characteristics of the channel perceived at the ports compared to the free-space channel. Consequently, conventional DoF analysis \cite{commode} based on closed-form free-space channel modeling and ideal, uncoupled antennas becomes invalid.

 To cope with the above challenges, we introduce the \emph{characteristic mode analysis} (CMA) theory to model the EM properties of antennas for DoF analysis. CMA is a theoretical framework that systematically investigates the inherent resonant properties of antenna structures by decomposing their surface currents into orthogonal \emph{characteristic modes}\cite{cmaoverview}. Eigenvalue equations derived from the \emph{electric field integral equation} are solved within CMA to quantify the EM properties for each characteristic mode\cite{cmatextbook}. 
In this way, the antenna's excitation properties are represented through the excitation coefficients of characteristic modes, while the antenna's radiation properties are modeled via the surface current (electric field) of characteristic modes.
Based on the analysis results provided by CMA, we integrate the EM property of antennas with the Green function-based modeling of the wireless propagation environment, giving rise to a \emph{CMA-modeled MIMO system}.
Moreover, we derive the \emph{achievable DoF} of MIMO based on the proposed CMA-based DoF analysis framework, serving as an evaluation of the multiplexing gain. 
Matrix analysis is also employed to decouple the antenna EM properties from the wireless propagation channel, studying how the antenna's EM properties affect the achievable DoF.

Our contributions are summarized below.
\begin{enumerate}
\item We propose a \emph{CMA-based DoF analysis framework}, where the CMA theory is applied to model the mapping relation between the input (output) EM signals and the surface current (electric field) of the antennas. The wireless propagation environment is modeled via the dyadic Green function. In the proposed CMA-modeled MIMO system, both the characteristics of antennas and the wireless propagation environment are considered.

\item Considering the EM characteristics of transceiver antennas, we derive the signal model and the the achievable DoF of the proposed CMA-modeled MIMO system. It is proven that the minimum among the DoF of the wireless propagation environment without the influence of the antennas, the number of antenna ports, and the number of characteristic modes provides an upper bound for the achievable DoF. Moreover, we demonstrate that the CMA-modeled MIMO system applies under general conditions, while the conventional MIMO adopting half-wavelength-spaced antenna arrays constitutes a special case of the proposed CMA-modeled MIMO.

\item We model the characteristic mode optimization problem of transceiver antennas in MIMO to maximize the achievable DoF. A case study where the MIMO system employs the reconfigurable holographic surfaces\footnote{The proposed DoF analysis framework can be extended to other types of antennas that can be modeled using CMA, i.e., antennas with metallic radiating structure.} (RHS)~\cite{rhspic} as antennas is discussed. \footnote{The reason for choosing RHS is that a change in the RHS configuration can affect the inherent structure and the EM property of RHS, thus affecting the achievable DoF.} We then propose a CMA-based genetic algorithm (CMA-GA) to optimize the RHS configuration for achievable DoF maximization. 
Simulation results show that the proposed CMA-GA algorithm can improve the achievable DoF of the RHS-deployed MIMO system by optimizing the EM properties of RHS, serving as a guideline for RHS reconfiguration.
\end{enumerate}

The rest of this paper is organized as follows. In Section \ref{sec::sys_mod}, the MIMO communication scenario is described, and the CMA theory is introduced, based on which the system model of the CMA-modeled MIMO is provided.  Section \ref{sec::dofa} analyzes the achievable DoF and discusses the relationship between the CMA-modeled MIMO system and the conventional MIMO system. Section \ref{sec::optimization} presents a characteristic mode optimization problem for maximizing the achievable DoF and a case study where the transceiver employs RHS as antennas is given. An RHS configuration optimization algorithm, CMA-GA, is then proposed. In Section \ref{sec::sim}, simulation results are provided, and conclusions are drawn in Section \ref{sec::conclusion}.

Throughout the paper, we use the following notation.
Vectors and matrices are represented by lower-case and upper-case boldface letters, respectively. The writing $\mathbf{X} \in \mathbb{C}^{a \times b}$ means that the size of $\mathbf{X}$ is ${a \times b}$ and each element of $\mathbf{X}$ is a complex number. 
We use $(\cdot)^T, (\cdot)^H$ and $(\cdot)^\dagger$
to denote the transpose, conjugate transpose, and pseudo-inverse operation, respectively. Besides, ${\rm diag}(x_1, ..., x_N)$ represents a diagonal matrix generated from  $x_1, ..., x_N$ and ${\rm blkdiag}(\mathbf{X}_1, ..., \mathbf{X}_N)$ represents a block diagonal matrix generated from matrices $\mathbf{X}_1, ..., \mathbf{X}_N$, and ${\rm j}$ denotes the imaginary unit.
\section{System Model}\label{sec::sys_mod}
In this section, we first describe a general MIMO communication scenario. Next, the signal model of the CMA-modeled MIMO communication is divided into three parts and elaborated successively, i.e., the CMA-based signal model of the transmit antenna, the dyadic Green Function-based channel model, and the CMA-based signal model of the receive antenna. 

\subsection{Scenario Description}
\begin{figure}[t]
	\centerline{\includegraphics[width=9cm,height=3.6cm]{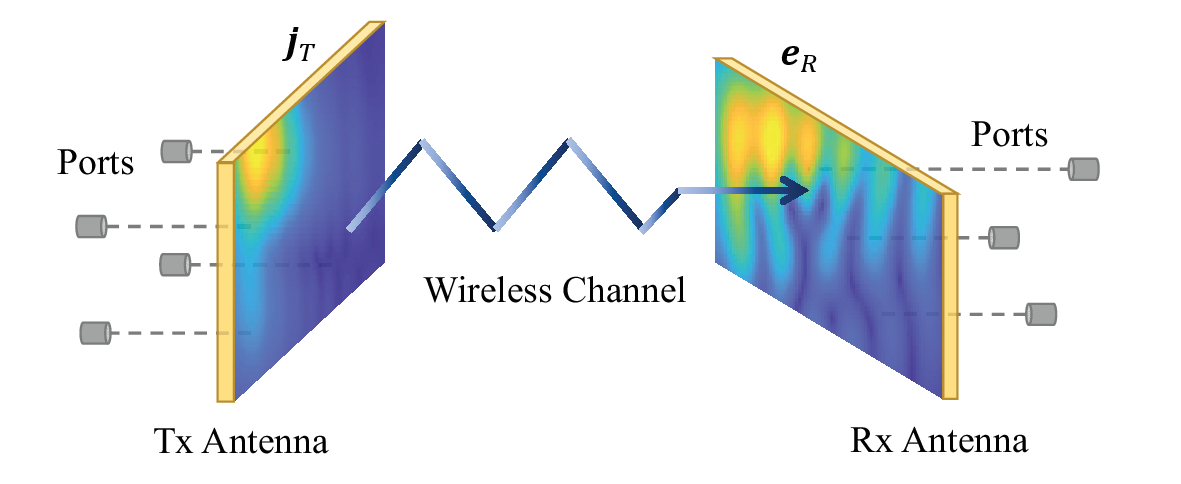}}
	\caption{An illustration of the discussed MIMO communication system.}
	\label{fig::scenario}
\end{figure}

As shown in Fig.~\ref{fig::scenario}, we consider the LoS MIMO communication system consisting of one transmit antenna with $L_T$ ports and one receive antenna with $L_R$ ports. At the transmitter side, the EM signal denoted by $\bm{s}_T \in \mathbb{C}^{L_T\times1}$ is fed into the antenna ports, and a current $\bm{j}_T$ is induced on the surface of the antenna. The surface current generates an EM wave, which propagates through the wireless environment and impinges on the aperture of the receive antenna, turning into the electric field $\bm{e}_{R}$. At the receiver side, the impinged electric field $\bm{e}_{R}$ results in the received EM signals $\bm{s}_R \in \mathbb{C}^{L_R\times1}$ at antenna ports. The entire propagation process from the transmitter to the receiver can be denoted as
\begin{equation}
	\bm{s}_{R} =\mathbf{U}_R\mathbf{G}\mathbf{U}_T\bm{s}_T+\bm{n},\label{equ::sigmodel}
\end{equation}
where $\mathbf{U}_T$ denotes the mapping relation between the input signal and the surface current $\bm{j}_T$, i.e., $\bm{j}_T = \mathbf{U}_T\bm{s}_T $. The term $\mathbf{G}$ denotes the propagation of EM waves in the free space, $\mathbf{U}_R$ denotes the mapping relation between the impinged electric field and the received signal, i.e., $\bm{s}_T= \mathbf{U}_R\bm{e}_R$. The term $\bm{n}$ denotes the additive white Gaussian noise. The terms $\mathbf{U}_T$ and $\mathbf{U}_R$ are subject to the EM property of the transmit and the receive antenna, which will be elaborated in detail using the CMA theory as follows.

\subsection{CMA-based Transmit Antenna Modeling}
\begin{figure}[t]
	\centerline{\includegraphics[width=7cm,height=3.5cm]{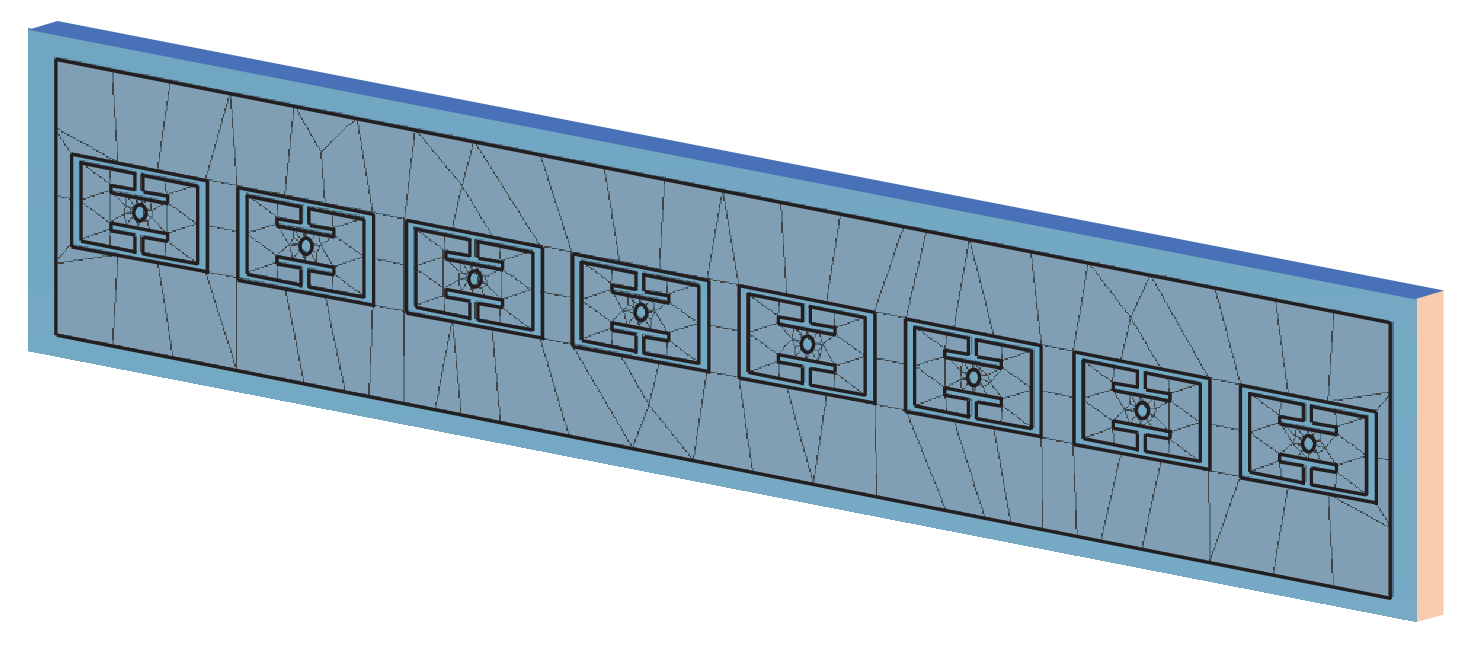}}
	\caption{The surface mesh of a patch antenna produced by MoM.}
	\label{fig::mesh}
\end{figure}

CMA is a versatile method to analyze the resonance properties of an arbitrary antenna by decomposing its impedance matrix into a number of orthogonal characteristic modes, which can assist the process of antenna design and analysis\cite{cmaoverview}. Each characteristic mode corresponds to a type of characteristic mode current distributed on the surface of the antenna and an induced characteristic mode electric field. Here, the CMA theory is applied to model the EM property of the transmit and receive antenna in the MIMO system\footnote{The following analysis, based on CMA, applies to both discrete antenna array and continuous aperture antenna\cite{patterndivision}.}.

CMA aims to implement antenna analysis by analyzing the \emph{electric field integral equation}, which gives the mapping relation between the surface current $\bm{j}_T$ and the induced electric field $\bm{e}_T$ on the antenna surface\cite{cmatextbook}. To calculate EFIE numerically, CMA adopts the \emph{method of moments} (MoM), based on which the surface of the antenna is divided into $N_T$ small triangular or quadrilateral elements, as illustrated in Fig.~\ref{fig::mesh}. 
By transforming EFIE, an impedance matrix $\mathbf{Z}_T$, which is determined by the structure of antenna, is built to depicted the relation between $\bm{e}_{T}$ and $\bm{j}_T$.
The impedance matrix can be split into the real and imaginary part as $\mathbf{Z}_T = \mathbf{R}_T+{\rm j}\mathbf{X}_T$, where $\mathbf{R}_T$ is related to the radiating EM field and $\mathbf{X}_T$ is related to the stored EM field\footnote{The electric field and the magnetic field of the stored EM field are out of phase by 90 degrees, which cannot propagate through long distances, and only exists in the reactive near-field region of the antenna\cite{antenna}.}. 

Next, a weighted eigenvalue equation can be built upon $\mathbf{R}_T$ and $\mathbf{X}_T$\cite{cmatextbook}, which is denoted as
\begin{equation}
	\mathbf{X}_T\tilde{\bm{j}}_{T,i}=\lambda_{T,i}\mathbf{R}_T\tilde{\bm{j}}_{T,i},
\end{equation}
where $\tilde{\bm{j}}_{T,i}$ and $\lambda_{T,i}$ are the $i$-th eigenvector and eigenvalue, and $\tilde{\bm{j}}_{T,i}$ has the physical meaning of the $i$-th characteristic mode current. Note that $\tilde{\bm{j}}_{T,i}$ is a vector field, hence, to facilitate subsequent analysis, the X, Y, Z polarization components of the $i$-th characteristic mode current at each divided element are combined to form a new vector $\bm{j}_{T,i}\in\mathbb{C}^{3N_T\times1}$.

Based on the orthogonality that stems from the eigenvalue function, it can be inferred that if one characteristic mode has strong currents at a specific point on the surface of the antenna, all other characteristic modes should have very weak currents at the same point. After performing a normalization on the surface current, such an orthogonal property can be denoted as
\begin{equation}
	\bm{j}_{T,l}^T\bm{j}_{T,i} = \delta_{l,i},
\end{equation}
where $\delta_{l,i}=1, l=i$ and $\delta_{l,i}=0, l\neq i$.

Therefore, for the transmit antenna, the surface current can be represented by a linear superposition of a series of characteristic mode currents as
\begin{equation}
	\bm{j}_T = \sum_{i=1}^{n_T}\sum_{l=1}^{L_T}\alpha_{T,i,l}\bm{j}_{T,i},
\end{equation}
where $n_T$ is the number of characteristic modes of the transmit antenna, $L_T$ is the number of ports of the transmit antenna, $\alpha_{T, i,l}$ is the superposition coefficient, and $\bm{j}_{T, i}$ is the $i$-th characteristic mode current. The dimension of the characteristic mode current is three times the number of divided elements on the antenna surface because the three polarization directions are considered.

The superposition coefficient can be denoted as\cite{cmatextbook}
\begin{equation}
	\alpha_{T,i,l} = m_{T,i}V_{T,i,l}s_{T,l},
\end{equation}
where $s_{T,l}$ is the input signal at the $l$-th port of the antenna, $m_{T,i}$ is the significance coefficient of $i$-th characteristic mode and characterizes each mode's radiation ability. The term $m_{T,i}$ is related to the eigenvalue $\lambda_{T,i}$ with the following equation as\cite{cmatextbook}
\begin{equation}
	m_{T,i} = \frac{1}{1+{\rm j}\lambda_{T,i}}.
\end{equation}
The term $V_{T, i,l}$ denotes the \emph{modal excitation coefficient}, which evaluates the ability of the $l$-th port of the antenna to excite the $i$-th characteristic mode.

\begin{figure}[t]
	\centerline{\includegraphics[width=9cm,height=5.2cm]{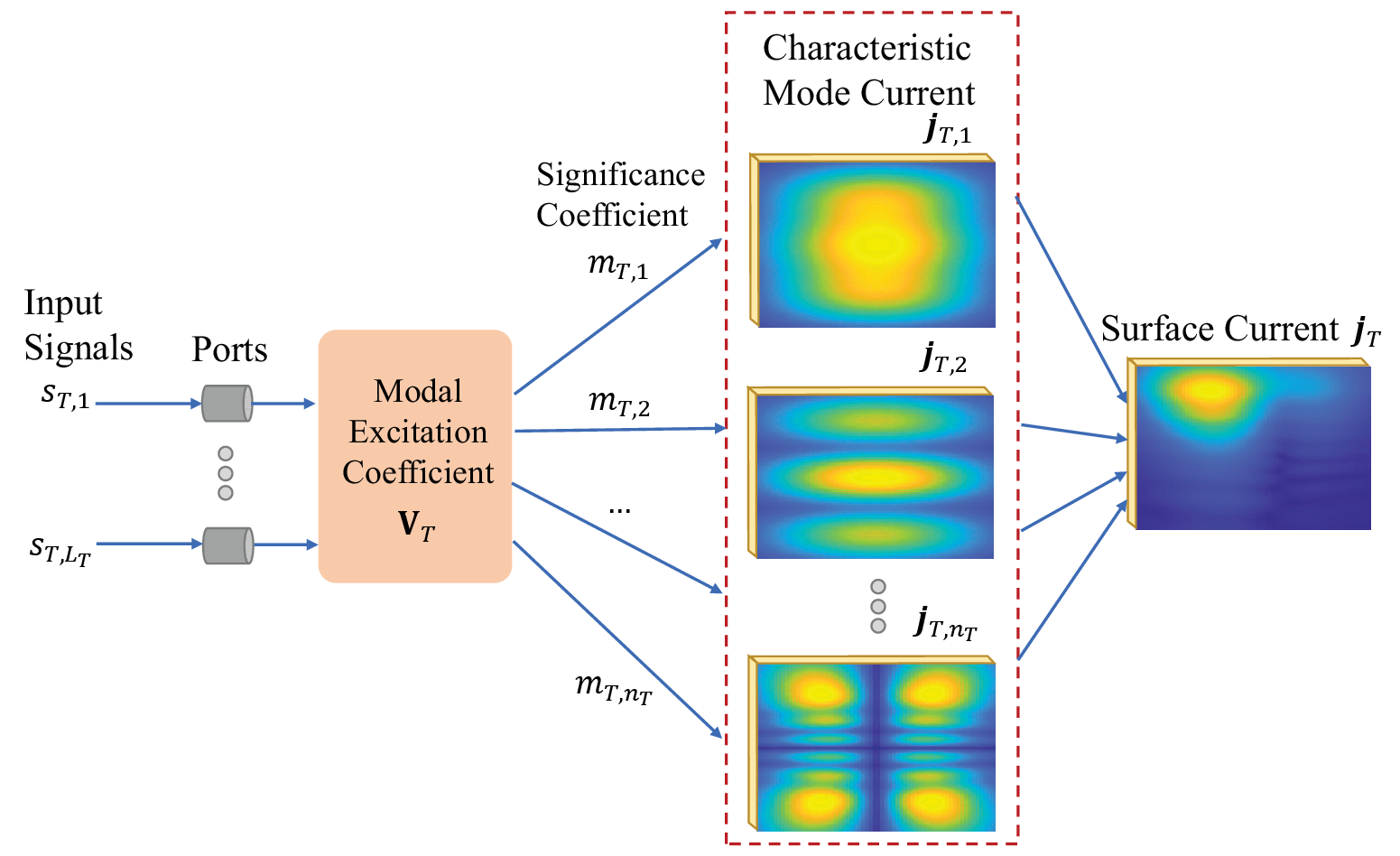}}
	\caption{An illustration of the transmitter antenna showing how the surface current is generated by the input EM signals based on CMA.}
	\label{fig::tx}
\end{figure}

As shown in Fig.\ref{fig::tx}, how the input signals from the ports transform into the surface current, which is denoted by $\mathbf{U}_T$ in (\ref{equ::sigmodel}), can be further expressed as
\begin{equation}
	\bm{j}_T= \overline{\mathbf{J}}_T {\rm diag}\{\bm{m}_T\}\mathbf{V}_T\bm{s}_T,\label{equ::jcma}
\end{equation}
where $\overline{\mathbf{J}}_T = [\bm{j}_{T,1},\bm{j}_{T,2},...,\bm{j}_{T,n_T}]\in\mathbb{C}^{3N_T \times n_T}$ concatenates all characteristic mode currents into a matrix, $\bm{m}_T\in \mathbb{C}^{n_T\times1}$ denotes a vector including all significance coefficients, $\mathbf{V}_T \in \mathbb{C}^{n_T \times L_T}$ denotes the modal excitation coefficient matrix. 
Here, equation (\ref{equ::jcma}) illustrates the radiation properties of the transmit antenna using the surface current and the significance coefficient of all characteristic modes, while the excitation properties is modeled using the modal excitation coefficient for all characteristic modes.

\subsection{Dyadic Green Function-based Channel Modeling}

The dyadic Green function is applied here to depict the LoS wireless channel in the free space, i.e., how the surface current on the transmit antenna induces the electric field on the surface of the receive antenna ~\cite{hmimocomm}. The propagation process is denoted as\footnote{The surface current $\bm{j}_T(\bm{r}')$, electric field $\bm{e}_R(r)$, and the dyadic green function $\overline{\overline{\mathbf{G}}}(\bm{r},\bm{r}')$ are continuous functions in (\ref{equ::green}) and (\ref{equ::dgreen}).}
\begin{equation}
	\bm{e}_R(r) = -{\rm j}\omega\mu_0\int_{A_T}\overline{\overline{\mathbf{G}}}(\bm{r},\bm{r}') \bm{j}_T(\bm{r}')d\bm{r}',\label{equ::green}
\end{equation}
where $\omega$ is the frequency of the EM wave, $\mu_0$ is the permittivity of the air, $\overline{\overline{\mathbf{G}}}$ denotes the dyadic Green function, the region $A_T$ denotes the surface of the transmit antenna, $\bm{r}$ and $\bm{r}'$ denote two arbitrary points on the surface of the receive antenna and the transmit antenna, respectively.

Under the LoS wireless propagation condition, the analytical expression of the dyadic Green function is obtained as\cite{electromagnetic}


\begin{align}
	&\overline{\overline{\mathbf{G}}}(\bm{r},\bm{r}')=-\frac{{\rm j}\eta\mathrm{exp}(-jk_0d)}{2\lambda d}[\left(\mathbf{I}-\widehat{\boldsymbol{d}}\cdot\widehat{\boldsymbol{d}}^T\right)+\nonumber\\
	&\frac{{\rm j}\lambda}{2\pi d}\left(\mathbf{I}-3\widehat{\boldsymbol{d}}\cdot\widehat{\boldsymbol{d}}^T\right) -\frac{\lambda^2}{(2\pi d)^2}\left(\mathbf{I}-3\widehat{\boldsymbol{d}}\cdot\widehat{\boldsymbol{d}}^T\right)],\label{equ::dgreen}
\end{align}
where $k_0=2\pi/\lambda$ denotes the wavenumber in the air, $\lambda$ is the wavelength, $\eta$ denotes the impedance of the free space, $\bm{d} = \bm{r}-\bm{r}'$ denotes the distance between the point $\bm{r}'$ on the transmit antenna and the point $\bm{r}$ on the receive antenna, $d = |\bm{r}-\bm{r}'|$ and $\widehat{\boldsymbol{d}}=\bm{d}/d$.

For the convenience of calculation, the continuous dyadic Green function $\overline{\overline{\mathbf{G}}}$ is discretized. To this end, the surfaces of the transmit and receive antennas are divided into $N_T$ and $N_R$ small elements, respectively, and each small element is considered to have a constant current or electric field. The same surface mesh is adopted in the CMA of the transmit and receive antennas to ensure consistency. Hence, the integral equation (\ref{equ::green}) can be rewritten as
\begin{equation}
	\bm{e}_{R} = \mathbf{G}\bm{j}_T,\label{equ::discGreen}
\end{equation}
where $\mathbf{G}$ is the discrete version of $\overline{\overline{\mathbf{G}}}$. The dimension of $\mathbf{G}$ is $3N_R \times 3N_T$ because of the three directions of polarization of EM waves. 

\subsection{CMA-Based Receive Antenna Modeling}
Based on the CMA theory, how the impinged electric field $\bm{e}_R  \in \mathbb{C}^{3N_R \times 1}$ generates the received signal $\bm{s}_R$ output from the receive antenna ports is given. Note that only a fraction of $\bm{e}_R$ can be appropriately received by the receive antenna because the residual fraction of $\bm{e}_R$, denoted by $\bm{e}_{R, res}$, is orthogonal to all the characteristic mode electric fields of the receive antenna, i.e., $\bm{e}_{R, res}$ either dissipates or is reflected. For example, for an X-polarized receive antenna, the characteristic modes of the receive antenna are all X-polarized and a Y-polarized electric field impinging on the antenna cannot be appropriately received. The decomposition of $\bm{e}_R$ can be denoted as

\begin{equation}
	\bm{e}_{R} = \sum_{i=1}^{n_R}{\beta_{R,i}m_{R,i}}\bm{e}_{R,i}+\bm{e}_{R,res},
\end{equation}
where $n_R$ is the number of characteristic modes of the receive antenna, $\beta_{R, i}$, $m_{R, i}$, $\bm{e}_{R, i}$ are the superposition coefficient, the significance coefficient, and the electric field of the $i$-th characteristic mode of the receive antenna, respectively. The term $m_{R,i}$ is given as $m_{R,i} = 1/(1+{\rm j}\lambda_{R,i})$, where $\lambda_{R,i}$ is the eigenvalue of $i$-th characteristic mode. 

\begin{lemma}\label{prop::rx} \rm
	The vector of superposition coefficient $\bm{\beta}_R=[\beta_{R,1},\beta_{R,2},...,\beta_{R,n_R}]^T$ is given as
	\begin{equation}
		\bm{\beta}_R={\rm diag}\{\bm{m}_{R}\}^{-1}\overline{\mathbf{E}}_R(\overline{\mathbf{E}}_R^H\overline{\mathbf{E}}_R)^{-1}\overline{\mathbf{E}}_R^H\bm{e}_{R}, \label{equ::rxbeta}
	\end{equation}
	where $\bm{m}_R\in\mathbb{C}^{n_R\times1}$ is a vector including all significance coefficients and  The term $\overline{\mathbf{E}}_R = [\bm{e}_{R,1},\bm{e}_{R,2},...,\bm{e}_{R,n_R}]\in\mathbb{C}^{3N_R\times n_R}$ is the matrix that collects all the electric fields of characteristic modes. 
\end{lemma}

\begin{proof}
	See Appendix~\ref{prop::rx}.
\end{proof}

The output EM signals $\bm{s}_R$ is related to the $\bm{\beta}_R$ via the matrix of modal excitation coefficient $\mathbf{V}_R$, which is denoted as
$\bm{\beta}_R = \mathbf{V}_R\bm{s}_R$. Since the number of characteristic modes $n_R$ does not necessarily equal the number of antenna ports, $\bm{s}_R$ is given as
\begin{equation}
	\bm{s}_R = (\mathbf{V}_R^H\mathbf{V}_R)^{-1}\mathbf{V}_R^H\bm{\beta}_R.\label{equ::rxbvs} 
\end{equation}

By combining (\ref{equ::rxbeta}) and (\ref{equ::rxbvs}), the process of how $\bm{e}_R$ generates $\bm{s}_R$, which is shown in Fig.\ref{fig::rx}, is given as
\begin{equation}
	\bm{s}_R=\mathbf{V}_R^{\dagger} {\rm diag}\{\bm{m}_R\}^{-1}\overline{\mathbf{E}}_R^\dagger \bm{e}_{R},\label{equ::ecma}
\end{equation}
where $(\cdot)^{\dagger}$ denotes the pseudo-inverse operator. Given the electric field and the significance coefficients of all characteristic modes, the radiation properties of the receive antenna are modeled via CMA. The modal excitation coefficient accounts for the antenna's excitation properties. By combining (\ref{equ::jcma}), (\ref{equ::discGreen}), and (\ref{equ::ecma}), the signal model for the CMA-modeled MIMO system is given as
\begin{equation}
	\bm{s}_R=\mathbf{V}_R^{\dagger} {\rm diag}\{\bm{m}_R\}^{-1}\overline{\mathbf{E}}_R^\dagger\mathbf{G}\overline{\mathbf{J}}_T {\rm diag}\{\bm{m}_T\}\mathbf{V}_T\bm{s}_T+\bm{n},\label{equ::CMAsignalmodel}
\end{equation}
where the LoS wireless channel is denoted with discretized dyadic Green function $\mathbf{G}$. The excitation and radiation properties of the transmit and receive antennas are characterized with CMA, i.e.,  $\overline{\mathbf{J}}_T {\rm diag}\{\bm{m}_T\}\mathbf{V}_T$ and $\mathbf{V}_R^{\dagger} {\rm diag}\{\bm{m}_R\}^{-1}\overline{\mathbf{E}}_R^\dagger$, respectively.

\begin{figure}[t]
	\centerline{\includegraphics[width=9.2cm,height=5.6cm]{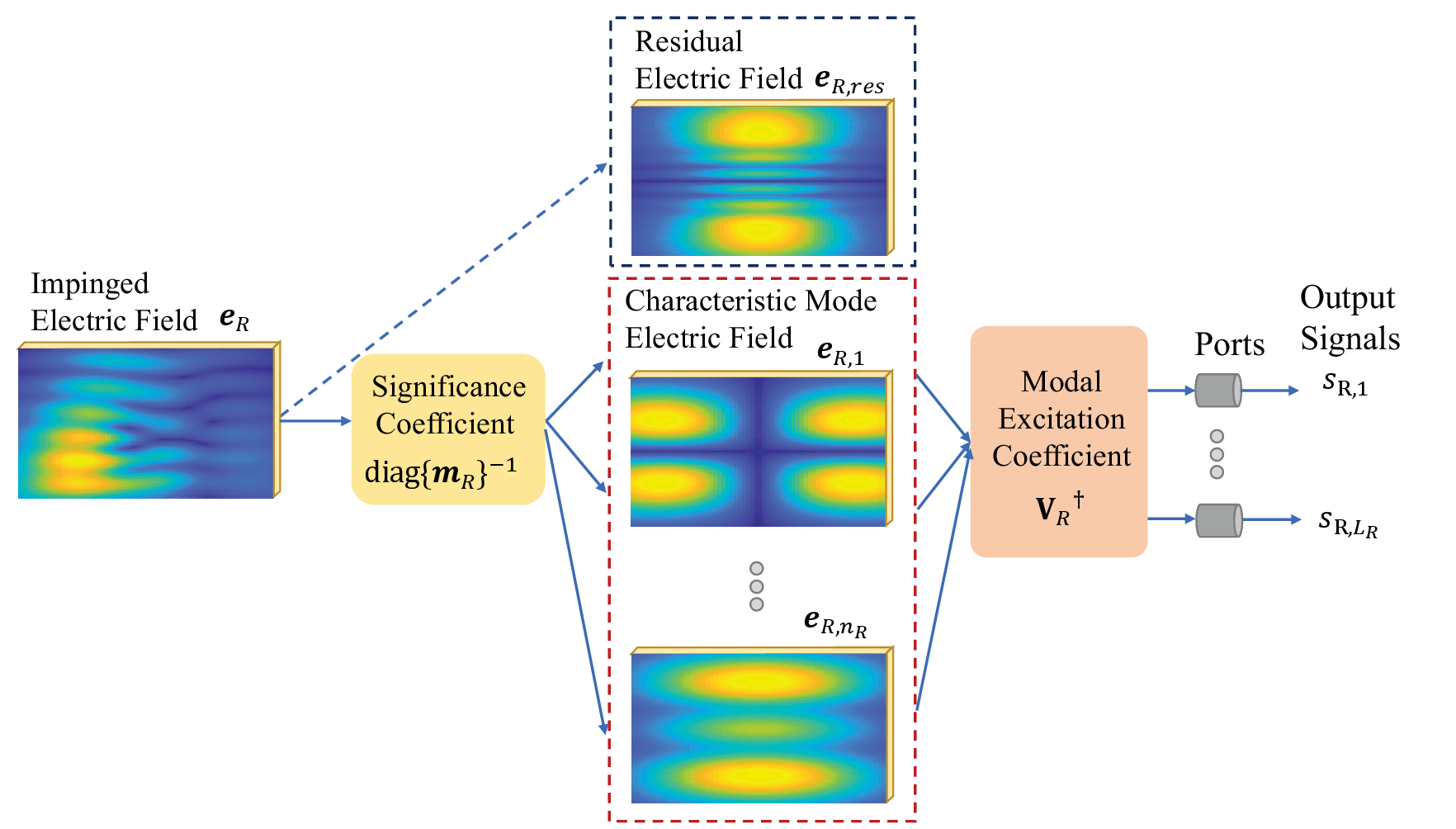}}
	\caption{An illustration of the receive antenna showing how the impinged electric field produces output signals based on CMA.}
	\label{fig::rx}
\end{figure}

\section{Degrees of Freedom Analysis} \label{sec::dofa}
In this section, we first present the degrees of freedom (DoF) based on the analysis of the dyadic Green function, and then analyze the DoF of the CMA-modeled MIMO communication system. Finally, we discuss the relationship between the CMA-modeled and conventional MIMO systems.

\subsection{Conventional Dyadic Green Function-Based DoF}
Most existing research analyzes the DoF of the MIMO system based solely on the wireless channel modeled using dyadic Green function\cite{commode,multidof,af}. Specifically, the dyadic Green function is approximated by the outer product of two sets of orthogonal continuous functions as
\begin{equation}
	\overline{\overline{\mathbf{G}}}(\bm{r},\bm{r}') \approx \sum_{n=1}^{D}\lambda_{n}v_{R,n}(\bm{r})\otimes v_{T,n}(\bm{r}')^H,\label{equ::svd}
\end{equation}
where $\otimes$ denotes the outer product operator, $v_{R,n}(\bm{r})$ is the $n$-th orthogonal electric field of the receive antenna, $v_{T,n}(\bm{r}')$ is the $n$-th orthogonal surface current of the transmit antenna, and $D$ denotes the number of orthogonal continuous functions\footnote{The term $D$ is set to satisfy a certain level of accuracy of decomposition, i.e., $\int_{A_T}\int_{A_R}{||\overline{\overline{\mathbf{G}}}(\bm{r},\bm{r}')-\sum_{n=1}^{D}\lambda_{n}v_{R,n}(\bm{r})\otimes v_{T,n}(\bm{r}')^H||^2d\bm{r}d\bm{r}'}\leq \epsilon$, where $\epsilon$ is the acceptable error, $A_R$ is the region of the surface of the receive antenna.}. Each pair of orthogonal functions forms a subchannel for independent data transmission. 

Based on the decomposition of the dyadic Green function, the degrees of freedom of the wireless channel are defined as the number of orthogonal continuous functions, i.e., $DoF_{\mathbf{G}} = D$. To numerically compute $DoF_{\mathbf{G}}$, $v_{R,n}(\bm{r})$ and $v_{T,n}(\bm{r}'), n=1,2,..., DoF_{\mathbf{G}}$, the singular value decomposition (SVD) is applied to the discretized dyadic Green function $\mathbf{G}$, which is denoted as

\begin{equation}
	\mathbf{G}=\mathbf{V}_T\mathbf{\Sigma}\mathbf{V}_R^H,\label{equ::svdg}
\end{equation}
where the singular vectors constituting $\mathbf{V}_T$ and $\mathbf{V}_R$ correspond to the discretized orthogonal surface current and the electric field.
In this way, $DoF_{\mathbf{G}}$ is given as
\begin{equation}
	DoF_{\mathbf{G}}={\rm Rank}(\mathbf{GG}^H).\label{equ::rankg}
\end{equation}

\subsection{Achievable DoF Analysis Based on CMA Theory}\label{subsec::adof}

DoF obtained via SVD on $\mathbf{G}$ in (\ref{equ::svdg}) gives the upper bound of the DoF of the MIMO system by assuming that the antennas can generate arbitrary distributions of electric field and surface current.
In reality, however, constrained by the physical structure and the EM theory, antennas cannot produce arbitrary current and electric field distributions, thus limiting the number of achievable DoF. 
Therefore, we propose the concept of the achievable DoF of the MIMO system, which takes the EM property of antennas into account. To this end, the CMA of antennas is applied to investigate the achievable DOF of the MIMO system.

\subsubsection{Achievable DoF Derivation}
To obtain the achievable DoF, we first define the equivalent channel $\tilde{\mathbf{H}}$, which depicts the coupling relationship between the input and output EM signals, considering the impact of the transceiver antennas. Based on the signal model in (\ref{equ::CMAsignalmodel}), $\tilde{\mathbf{H}}$ can be given as
\begin{equation}
	\tilde{\mathbf{H}} = \mathbf{V}_R^{\dagger} {\rm diag}\{\bm{m}_R\}^{-1}\overline{\mathbf{E}}_R^\dagger\mathbf{G}\overline{\mathbf{J}}_T {\rm diag}\{\bm{m}_T\}\mathbf{V}_T,\label{equ::equH}
\end{equation}
based on which the achievable DoF of the MIMO system is given as
\begin{equation}
	DoF_{\mathbf{H}} ={\rm Rank}(\mathbf{\tilde{H}\tilde{H}}^H).\label{equ::dofdef}
\end{equation}
\begin{proposition}\label{prop::dofsimp}
	\rm The achievable DoF can be simplified as 
\begin{equation}
	DoF_{\mathbf{H}} ={\rm Rank}(\mathbf{V}_R^{\dagger} {\rm diag}\{\bm{m}_R\}^{-1}\bm{\Gamma} {\rm diag}\{\bm{m}_T\}\mathbf{V}_T),\label{equ::dofsimp}
\end{equation}
where $\bm{\Gamma} \in \mathbb{C}^{n_R\times n_T}$ and $[\bm{\Gamma}]_{ij}$ represents the coupling strength between the electric field of the $i$-th characteristic field of the receive antenna and the surface current of the $j$-th characteristic field of the transmit antenna. The term $\bm{\Gamma}$ is obtained via a decomposition to the discrete wireless channel $\mathbf{G}$ based on the electric field and current of the characteristic modes, i.e., $\overline{\mathbf{E}}_R$ and $\overline{\mathbf{J}}_T$, which is denoted as
\begin{equation}
	\mathbf{G} = \overline{\mathbf{E}}_R\bm{\Gamma}\overline{\mathbf{J}}_T^T.\label{equ::decompG}
\end{equation}
\end{proposition}
\begin{proof}
	See Appendix \ref{app::prop4}.
\end{proof}

\subsubsection{Achievable DoF Analysis}
\label{rmk::gamma} Based on (\ref{equ::dofsimp}), the following analysis on $DoF_{\mathbf{H}}$ is given.

\begin{remark}
The achievable DoF of the MIMO system $DoF_{\mathbf{H}}$ is restricted both by the number of ports of the transceiver antennas and the number of characteristic modes of the transceiver antennas. Hence, we have
\begin{equation}
DoF_{\mathbf{H}} \leq \min\{L_T,L_R,n_T,n_R\}.
\end{equation}
Therefore, the number of data streams available for independent data transmission is no larger than the number of ports of antennas.
\end{remark}

\begin{proposition}\label{prop::dofrelation}
	\rm The relationship between the achievable DoF of the MIMO system, i.e., $DoF_{\mathbf{H}}$, and the DoF of the free space LoS wireless channel, i.e., $DoF_{\mathbf{G}}$, can be denoted as
	\begin{equation}
		DoF_{\mathbf{H}} \leq DoF_{\mathbf{G}},\label{equ::dofleq}
	\end{equation}
	The achievable DoF of the MIMO system also satisfies 
	\begin{equation}
		DoF_{\mathbf{H}} \geq {\rm Rank(\mathbf{V}_R)}+{\rm Rank(\mathbf{V}_T)}+{\rm Rank(\mathbf{\Gamma})}-n_R -n_T,\label{equ::dofgeq}
	\end{equation}
\end{proposition}
\begin{proof}
	See Appendix \ref{app::prop1}.
\end{proof}

The inequality (\ref{equ::dofleq}) implies that the DoF given by the SVD of $\mathbf{G}$ gives an upper bound on the achievable DoF of the MIMO system. This is because the influence of the transceiver antennas may reduce the rank of the equivalent channel matrix, as reflected in (\ref{equ::equH}). Besides, it can be inferred from (\ref{equ::dofgeq}) that if antenna ports are effective in exciting a large number of characteristic modes, leading to a high rank of the modal excitation coefficient matrix $\mathbf{V}_R$ and $\mathbf{V}_T$, a large achievable DoF can be obtained. Also, to increase $DoF_{\mathbf{H}}$, the rank of $\mathbf{\mathbf{\Gamma}}$ should be maximized. 

Moreover, the difference between the SVD of $\mathbf{G}$ in (\ref{equ::svd}) and the decomposition based on characteristic modes in (\ref{equ::decompG}) is discussed here. Unlike the singular value matrix $\mathbf{\Sigma}$, $\bm{\Gamma}$ may not be a diagonal matrix or a full-rank matrix, i.e., ${\rm Rank}(\mathbf{\Gamma}) \leq \min \{n_R,n_T\}$. This is because the characteristic modes are related to the impedance matrix of the transceiver antennas, i.e., $\mathbf{Z}_T$ and $\mathbf{Z}_R$, which are determined by the antennas' inherent property, e.g., radiating structure, material, and element spacing. Hence, the current and electric field of characteristic modes are not arbitrarily generated like $v_{R,n}(\bm{r})$ and $v_{T,n}(\bm{r}')$ obtained from the SVD of $\mathbf{G}$. A potential mismatch between the characteristic modes between the transceiver antennas may also lead to the rank deficiency of $\bm{\Gamma}$.


\subsection{Relationship between CMA-Aware MIMO and Conventional MIMO}

\begin{figure}[t]
	\centerline{\includegraphics[width=6.5cm,height=3.4cm]{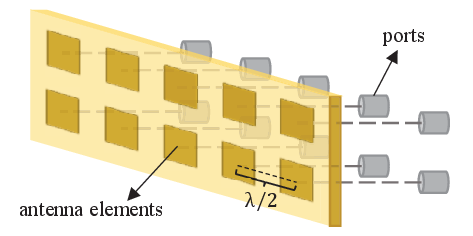}}
	\caption{An illustration of the discrete antenna array in the conventional MIMO system, where each antenna element is separately connected to a port.}
	\label{fig::conventionalmimo}
\end{figure}

Here, the relationship between the CMA-modeled MIMO and the conventional MIMO is discussed to unveil that the conventional MIMO, i.e. antenna elements spaced at half a wavelength\cite{convmimo}, is a special case of the proposed CMA-modeled MIMO. Specifically, by utilizing the property of conventional MIMO, the signal model of the CMA-modeled MIMO in (\ref{equ::CMAsignalmodel}) can be converted to that of the conventional MIMO, as illustrated below.

 For the discrete antenna array adopted in the conventional MIMO system, as illustrated in Fig.~\ref{fig::conventionalmimo}, the number of ports is the same as that of antenna elements. In this way, the surface current can be rewritten as the concatenation of the surface current of each antenna element as 
	\begin{equation}
		\bm{j}_T = [\bm{j}_{e,T,1},\bm{j}_{e,T,2},...,\bm{j}_{e,T,L_T}],
	\end{equation}
	where $\bm{j}_{e, T, l} \in \mathbb{C}^{3N_{e, T}\times1}$ denotes the surface current of $l$-th antenna element, and $N_{e, T} = \frac{N_T}{L_T}$ is the number of divided surface element of each transmit antenna element. 
	
\begin{lemma} \label{lem::txu}
\rm 
For the transmitter of the conventional MIMO system, the discrete antenna array is equipped with identical antenna elements and the element spacing is $\lambda/2$. In this way, the mapping relation $\mathbf{U}_T$ between $\bm{s}_T$ and $\bm{j}_T$ can be converted from the CMA-based modeling, i.e., $\mathbf{U}_T = \overline{\mathbf{J}}_T {\rm diag}\{\bm{m}_T\}\mathbf{V}_T$, to
\begin{equation}
\mathbf{U}_T ={\rm blkdiag}\{\overbrace{\hat{\bm{j}}_{e,T},\hat{\bm{j}}_{e,T},...,\hat{\bm{j}}_{e,T}}^{L_T}\},\label{equ::ut1}
\end{equation}
where $\hat{\bm{j}}_{e,T}\in \mathbb{C}^{3N_{e, T}\times1}$ is the normalized surface current of an antenna element when it is excited with $s_T = 1$.
\end{lemma}

\begin{proof}
A half-wavelength element spacing is adopted for the discrete antenna array, as illustrated in Fig.~\ref{fig::conventionalmimo}, leading to a negligible mutual coupling effect between antenna elements\cite{mc}. Hence, 
each antenna element is considered independently excited by the signal fed from the port connecting with the antenna element. Based on (\ref{equ::jcma}), $\bm{j}_{e,T,l}$ is denoted as
\begin{equation}
	\bm{j}_{e,T,l}= \overline{\mathbf{J}}_{T,l} {\rm diag}\{\bm{m}_{T,l}\}\bm{V}_{T,l} s_{T,l},
\end{equation}
where $\overline{\mathbf{J}}_{T,l}$, $\bm{m}_{T,l}$, and $\bm{V}_{T,l}$ denote the characteristic mode current, significance coefficient, and modal weighting coefficient of the $l$-th antenna element.
Since multiple characteristic modes cannot be effectively manipulated for each antenna element with only one port of input signal, $\bm{j}_{e, T,l}$ can be further reduced as
\begin{equation}
	\bm{j}_{e,T,l} = \hat{\bm{j}}_{e,T,l} s_{T,l},\label{equ::elementc}
\end{equation}
where $\hat{\bm{j}}_{e,T,l}$ is the normalized surface current of the $l$-th antenna element when $s_{T,l}=1$.  If the antenna elements are identical to each other, we have $\hat{\bm{j}}_{e, T,1}=,...,=\hat{\bm{j}}_{e, T, L_T}=\hat{\bm{j}}_{e, T}$. By rewriting (\ref{equ::elementc}) into the matrix form as $\bm{j}_T = \mathbf{U}_T\bm{s}_T$, (\ref{equ::ut1}) can be obtained, completing the proof.
\end{proof}

\begin{lemma}\label{lem::rxu}
At the receiver side, the mapping relation $\mathbf{U}_R$ between the output signals $\bm{s}_R$ and the impinged electric field $\bm{e}_R$ can be converted from $\mathbf{U}_R =\mathbf{V}_R^{\dagger} {\rm diag}\{\bm{m}_R\}^{-1}\overline{\mathbf{E}}_R^\dagger$ to
\begin{equation}
	\mathbf{U}_R = {\rm blkdiag}\{\overbrace{\hat{\bm{e}}_{e,R},\hat{\bm{e}}_{e,R},...,\hat{\bm{e}}_{e,R}}^{L_R}\}^\dagger,\label{equ::ur1}
\end{equation}
where $\hat{\bm{e}}_{e,R} \in \mathbb{C}^{3N_{e,R}\times1},l=1,2...,L_R,$ denotes the normalized electric field of each receive antenna element when $s_R=1$. The term $N_{e, R} = \frac{N_R}{L_R}$ is the number of divided surface elements of each receive antenna element.
\end{lemma}
\begin{proof}
A proof similar to the proof of Proposition~\ref{lem::txu} can be given, on account of the following properties of the discrete antenna array at the receiver side, i.e., $\lambda/2$ element spacing, identical antenna elements, and the independent connection between each antenna element and each port.
\end{proof}

	
Moreover, in the channel modeling of the conventional MIMO system, each antenna element is considered a point source for the transceiver discrete antenna arrays\cite{ps}. Hence, instead of applying the continuous Green function based on the EM theory, the wireless channel is modeled by a point-to-point version $\mathbf{\tilde{G}}\in\mathbb{C}^{L_R \times L_T}$, where $[\mathbf{\tilde{G}}]_{ij}$ denotes the channel between the $i$-th receive antenna element and the $j$-th transmit antenna element. 
	
By considering each antenna element as a point source, each antenna element's surface current and electric field are also reduced from vectors, i.e.,  $\hat{\bm{j}}_{e, T}$ and $\hat{\bm{e}}_{e, R}$, to complex scalars. Such scalars denote the element gain and are represented by $\rho_T$ and $\rho_R$ for the transceiver discrete antenna arrays, respectively. Therefore, $\mathbf{U}_T$ and $\mathbf{U}_R$ can be reduced from \textbf{Lemma}~\ref{lem::txu} and \textbf{Lemma}~\ref{lem::rxu} to 
\begin{align}
	&\mathbf{U}_T = \rho_T\mathbf{I}_{L_T},\label{equ::simput}\\
	&\mathbf{U}_R = \rho_R\mathbf{I}_{L_R}.\label{equ::simpur}
\end{align}
\begin{proposition}
Considering the property of discrete antenna arrays adopted in the conventional MIMO and the simplified point-to-point channel model, the signal model of the CMA-modeled MIMO system can be rewritten as 
\begin{equation}
	\bm{s}_R=\rho_T\rho_R\mathbf{\tilde{G}}\bm{s}_T+\bm{n},\label{equ::MIMOsm}
\end{equation}
which is in accordance with the signal model of the conventional MIMO. 
\end{proposition}
\begin{proof}
Note that for discrete antenna arrays in the conventional MIMO system, no mutual coupling effect is considered, and the excitation of each antenna element is independent of other elements. Besides, the point-to-point channel model is assumed, and thus, we have the simplified $\mathbf{U}_T$ and $\mathbf{U}_R$ given in (\ref{equ::simput}) and (\ref{equ::simpur}). Substituting the point-to-point channel model $\mathbf{\tilde{G}}$ and the the converted $\mathbf{U}_T$ and $\mathbf{U}_R$ into (\ref{equ::CMAsignalmodel}), (\ref{equ::MIMOsm}) is obtained.
\end{proof}
By comparing the signal model of the CMA-modeled MIMO and the conventional MIMO, i.e., (\ref{equ::CMAsignalmodel}) and (\ref{equ::MIMOsm}), it can be concluded that the CMA-modeled MIMO system is a more general representation by considering EM properties of both the antennas and wireless channel. 
This further verifies that, $DoF_{\mathbf{H}}$ is a precise evaluation of the DoF of the MIMO system, since it stems from the general CMA-modeled signal model.
	
	\section{Achievable DoF Maximization based on CMA Theory}\label{sec::optimization}
	In this section, we formulate the characteristic optimization problem to maximize the achievable DoF for the CMA-modeled MIMO system. Next, a case study is presented, where the transceiver employs a type of two-dimensional reconfigurable antenna, e.g., RHS, as antennas. A CMA-based genetic algorithm is proposed to solve the achievable DoF maximization problem of the case study by optimizing the RHS configuration, which serves as a guideline for the RHS configuration design.
	
	\subsection{Problem Formulation}
	Based on the analysis of achievable DoF in Section \ref{subsec::adof}, the property of characteristic modes has an influence on $DoF_{\mathbf{H}}$. Specifically, given the fixed free-space LoS wireless channel $\mathbf{G}$, the parameters of characteristic modes, including the characteristic mode current and electric field $\overline{\mathbf{J}}_T,\overline{\mathbf{E}}_R$, the significance coefficient $\bm{m}_T,\bm{m}_R$, and the modal weighting coefficients $\mathbf{V}_T,\mathbf{V}_R$, should be optimized to improve the achievable DoF. The achievable DoF maximization problem can be formulated as
	\begin{equation}
		\max_{\mathbf{V}_R, \bm{m}_R, \overline{\mathbf{E}}_R,\overline{\mathbf{J}}_T \bm{m}_T,\mathbf{V}_T}DoF_{\mathbf{H}}.\label{equ::op1}
	\end{equation}
	
	Note that the parameters of characteristic modes are determined by the inherent property of the transceiver antennas, as stated in Section~\textbf{\ref{rmk::gamma}}. Hence, the radiating structure of antennas should be optimized to maximize $DoF_{\mathbf{H}}$. 
	
	To effectively illustrate the proposed CMA-modeled MIMO system and present the optimization of characteristic modes, we consider a type of antenna called \emph{reconfigurable antennas}\cite{reconfigurableantenna}. Elements of reconfigurable antennas are equipped with active components, e.g., PIN diodes or varactors, to flexibly adjust the radiating structure of antennas to modify the EM property.   
	In this way, the characteristic modes of antennas can be reconfigured to fit the varying wireless propagation environment to solve (\ref{equ::op1}).
	
	\subsection{Case Study: Reconfigurable Holographic Surface}\label{subsec::rhs}
	
	\begin{figure}[t]
		\centerline{\includegraphics[width=9cm,height=5.7cm]{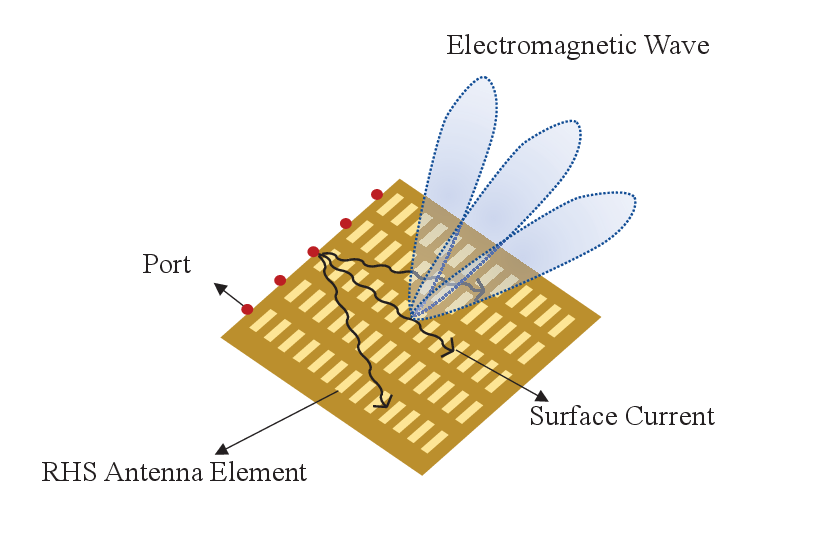}}
		\caption{The illustration of an RHS antenna.}
		\label{fig::rhs}
	\end{figure}
	
	It is noted that the antennas discussed in the CMA-modeled MIMO system can either have a conventional structure, where each antenna element is connected to one port, as illustrated in Fig.\ref{fig::conventionalmimo}, or have a hybrid structure when operating at high frequencies, where multiple antenna elements are excited by a single port.
For the sake of illustration, here we take a case of a reconfigurable antenna adopting the hybrid structure, e.g., RHS, as an example. By optimizing the achievable DoF, RHS-assisted MIMO communication with effective multiplexing can be realized, as shown below.
	
	As illustrated in Fig.~\ref{fig::rhs}, RHS is a type of metamaterial-based reconfigurable antenna that incorporates densely packed sub-wavelength metamaterial antenna elements\cite{rhs2}. Each antenna element is deployed with pin diodes to control the element's state. Owing to the changeable configuration of all antenna elements' states, the characteristic modes of the RHS can be controlled, further affecting the excited and received EM waves. Different from reconfigurable intelligent surfaces, the ports of the RHS are integrated with the meta-surface to generate EM waves propagating along the meta-surface and exciting the RHS elements one by one\footnote{The power of input signal at each port of RHS can be optimized to improve spectrum efficiency but is irrelevant to the achievable DoF.}.
	
	Here we consider the RHS at the transmitter and the receiver to have $L_T$ and $L_R$ ports, respectively, and each port corresponds to $N_{L, T}$ and $N_{L, R}$ antenna elements. Therefore, the transmit and receive RHS configuration can be denoted by $\bm{\phi}_T\in\mathbb{Z}^{L_T N_{L, T}\times 1}$ and $\bm{\phi}_R\in\mathbb{Z}^{L_R N_{L, R}\times 1}$, respectively, and each element in $\bm{\phi}_T$ and $\bm{\phi}_R$ is a 0-1 integer variable and denotes the state of the PIN diode. With RHS serving as antennas, the characteristic mode optimization problem in (\ref{equ::op1}) is reformulated as:
	\begin{align}
		&\max_{\bm{\phi}_T,\bm{\phi}_R}DoF_{\mathbf{H}}. \label{equ::op2}\\
		s.t.\ &\bm{\phi}_T\in \{0,1\}^{L_T N_{L,T}\times1},\tag{\ref{equ::op2}{a}}\\
		&\bm{\phi}_R\in \{0,1\}^{L_T N_{L,R}\times1},\tag{\ref{equ::op2}{b}}\\
		& \overline{\mathbf{J}}_T, \mathbf{V}_T, \bm{m}_T = f_T(\bm{\phi}_T),\tag{\ref{equ::op2}{c}}\\
		& \overline{\mathbf{E}}_R, \mathbf{V}_R, \bm{m}_R = f_R(\bm{\phi}_R),\tag{\ref{equ::op2}{d}}
	\end{align}
	where (\ref{equ::op2}{a}) and (\ref{equ::op2}{b}) denote the 1-bit constraints of the transmitter and the receiver RHS configuration. Functions $f_T(\cdot)$ and $f_R(\cdot)$ in (\ref{equ::op2}{c}) and (\ref{equ::op2}{d}) denote the mapping relation between the parameters of characteristic modes and the configuration of the transmitter and receiver RHS, respectively.
	
	\subsection{Algorithm Design}
	To solve (\ref{equ::op2}), we design a CMA-based genetic algorithm (CMA-GA) to optimize the  configuration of transmitter and receiver RHS, i.e.,  $\bm{\phi}_T$ and $\bm{\phi}_R$, considering that the mapping relation between the RHS configuration and the parameters of characteristic modes
	is nontrivial and challenging to be analytically obtained due to the complicated RHS structure.
	The proposed CMA-GA achievable DoF maximization algorithm is summarized in \textbf{Algorithm 1}. 
	
	\subsubsection{Initialization} The CMA-GA is first initialized in Steps 1-2. In Step 1, a number of RHS configurations are randomly generated, and the parameters of characteristic modes corresponding to these RHS configurations are obtained using CMA simulation tool in the CST Microwave Studio, an EM full-wave simulation software.
	
	In Step 2, the initial population of the CMA-GA is generated, which consists of $I$ transceiver RHS configurations $\bm{\phi}$. The term $\bm{\phi}$ is defined as the concatenation of the transmit and receive RHS configuration, which is denoted as
	\begin{equation}
		\bm{\phi} = [\bm{\phi}_T;\bm{\phi}_R],
	\end{equation}
	where $ \bm{\phi} \in \{0,1\}^{(L_T N_{L,T}+L_R N_{L,R})\times 1}$ denotes the variable to be optimized.
	\subsubsection{Generation}
	After the initialization process, $K$ generations are executed, as described from Step 3 to Step 9, in which the transceiver RHS configurations $\bm{\phi}$ undergo the crossover and mutation process for the population to evolve. Finally, $\bm{\phi}$ with the highest fitness score is generated so that the achievable DoF of the MIMO system is maximized.
	
	Here we define the fitness function for calculating the fitness score in CMA-GA. The physical meaning of $DoF_{\mathbf{H}}$ is the number of subchannels obtained from decomposing the equivalent channel matrix, and each subchannel corresponds to a singular value, which evaluates the data transmission ability. In practice, given the limited signal-to-noise ratio (SNR), subchannels with small singular values cannot support effective data transmission. Therefore, given the SVD result, the achievable DoF can be calculated as \cite{commode}
	\begin{equation}
		\sigma_1^2 \geq \sigma_2^2...\geq \sigma_{DoF_{\mathbf{H}}}^2\geq \gamma\sigma_1^2 ,\label{equ::numdof}\\
	\end{equation}
	where $\sigma_{l}$ is the $l$-th singular value of SVD on the equivalent channel $\tilde{\mathbf{H}}$, $\gamma$ is a predefined threshold related to SNR, and $L_m = \min\{L_T, L_R\}$ denotes the number of singular values.
	As inferred from (\ref{equ::numdof}), to increase the achievable DoF of the MIMO system, the singular values of the equivalent channel matrix should be evenly distributed. To this end, the \emph{fitness function} for evaluating the transceiver RHS configuration $\bm{\phi}$ is set as the negative value of the standard deviation of singular values, which is denoted as
	\begin{equation}
		f_{CMA}(\bm{\phi}) = -\sqrt{\frac{1}{L_m}\sum_{l=1}^{L_m}(\sigma_{l}-\overline\sigma)^2}, \label{equ::fit}
	\end{equation}
	where $\overline\sigma$ is the average singular value. A larger fitness score indicates that a more balanced distribution of singular values is obtained, which leads to a higher $DoF_{\mathbf{H}}$.

	
	In Step 3, $I_p$ configurations are selected from the population as the parent configuration based on probability and fitness scores\cite{geneticalgo}. In Step 4, given the chosen parent configurations, the child configurations are generated via the crossover and mutation processes. 
	In Step 5-7,  the parameters of characteristic modes are first obtained by utilizing CMA simulation tool, which are then used to calculate the equivalent channel via (\ref{equ::equH}). Next, the fitness score for each child configuration is calculated based on (\ref{equ::fit}). 
	In Step 8, the population is updated by retaining configurations with a high fitness score.
	
	Finally, after the iteration of generations is finished, the transceiver RHS configuration with the largest fitness score is chosen, representing the optimized configuration to produce the maximized $DoF_{\mathbf{H}}$. 
	\begin{algorithm}[t]
		\label{alg::cmaga}
		
		\caption{CMA-based Genetic Algorithm} 
		\vspace*{0.02in}\hspace*{0.02in} {\bf Input:} The structure of transceiver RHS antennas, the wireless channel $\mathbf{G}$, the number of generations $K$, the population size $I$. 
		\begin{algorithmic}[1]
			\State \textbf {Step 1.} Obtain parameters of characteristic modes of random RHS configurations with CMA tool in CST.
			\State \textbf{Step 2.} Initialize the population with $I$ random transceiver RHS configurations and calculate their fitness scores.
			\For{generation $= 0, 1, ..., K$}
			\State \textbf{Step 3.} Select $I_p$ parent configurations from the population based on probability.
			\State \textbf{Step 4.} Generate $I_p$ child configurations via crossover between each pair of parent configurations and mutation.
			\State \textbf{Step 5.} Obtain $\mathbf{V}_R, \bm{m}_R, \overline{\mathbf{E}}_R,\overline{\mathbf{J}}_T \bm{m}_T,\mathbf{V}_T$ for each child configuration with CMA tool in CST.
			\State \textbf{Step 6.} Calculate the equivalent channel $\tilde{\mathbf{H}}$ for each child configuration based on (\ref{equ::equH}).
			\State \textbf{Step 7.} Calculate the fitness score for each child configuration based on (\ref{equ::fit}). 
			\State \textbf{Step 8.} Update the population according to the fitness scores.
			\EndFor
			
		\end{algorithmic}
		\hspace*{0.02in} {\bf Output:} 
		Configuration of transceiver RHS antennas with the largest fitness score.
	\end{algorithm}
	
	\section{Simulation Results}\label{sec::sim}
\begin{figure}[t]
	\centerline{\includegraphics[width=5cm,height=5cm]{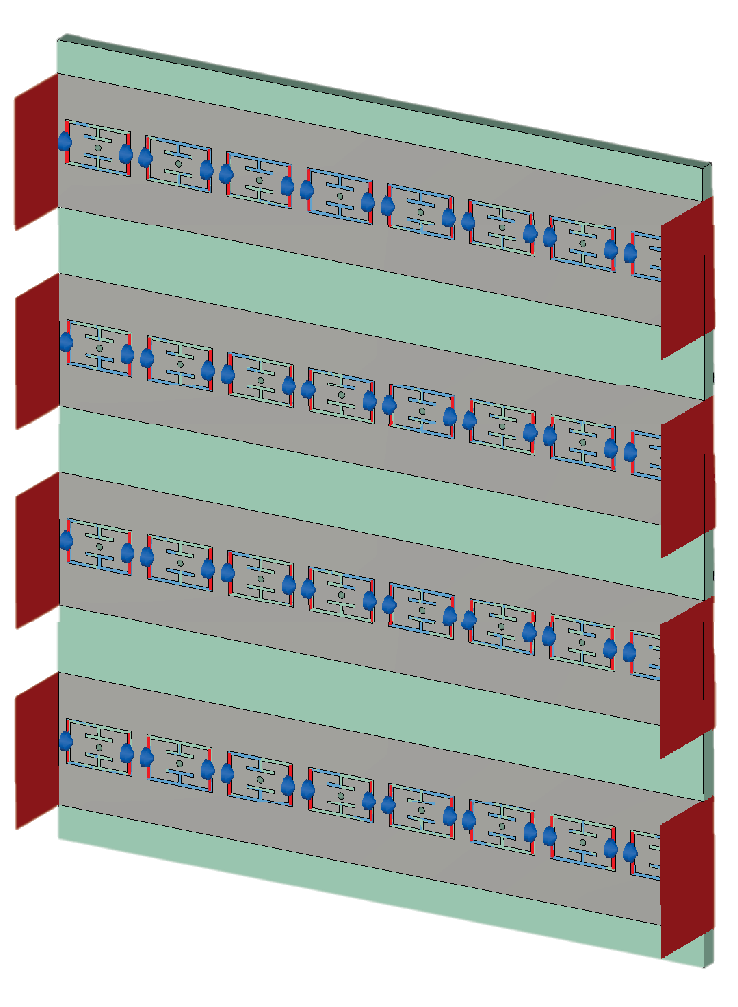}}
	\caption{ An illustration of the considered RHS antenna with $L_T=4$  and $N_{L,T}=8$.}
	\vspace{3mm}
	\label{fig::rhsp}
\end{figure}
		\begin{figure*}[th]
	\centerline{\includegraphics[width=20cm,height=4cm]{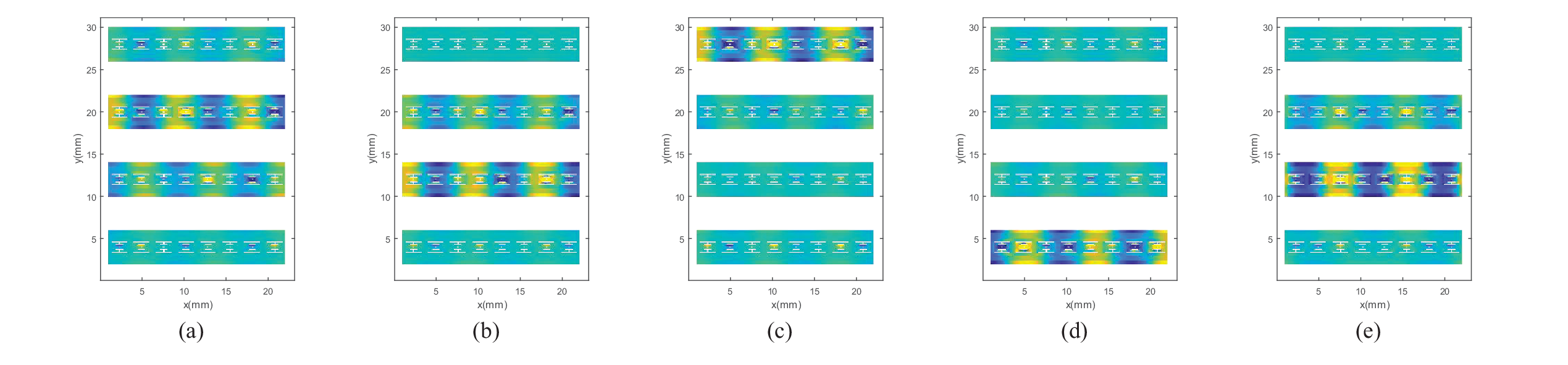}}
	\caption{The surface current of five characteristic modes of the RHS antenna with the strongest radiation ability.
	 }
	\label{fig::jcma}
\end{figure*}

In this section, we present simulation results of a MIMO communication system equipped with two RHS antennas serving as the transmit and receive antennas, respectively. The simulated result of different characteristic modes of the RHS antenna is demonstrated first. Next, the performance of the proposed CMA-GA for the achievable DoF maximization is evaluated. The relationship between the achievable DoF and the number of data streams utilized for communication is demonstrated.
	
For the simulation setup, we consider the case where the RHS antennas\cite{RHS} work at 27 GHz and the channel is the LoS channel modeled by (\ref{equ::dgreen}). The number of antenna elements for each port of RHS is set as $N_{L, T}=N_{L, R}=8$. The number of RHS ports is $3 \leq L_T, L_R \leq 7$. An RHS antenna with 4 ports is illustrated in Fig.~\ref{fig::rhsp}, where two pin diodes are integrated on each RHS antenna element and the red rectangles on two sides of the RHS denote the waveguide ports. One side of the waveguide ports are used for signal transmission and reception and the other for matching. The element spacing of RHS is set as $0.24\lambda$ horizontally and $0.72\lambda$ vertically. The distance between the transmit antenna and the receive antenna is set within the near-field range of the MIMO system, i.e., $2(D_T+D_R)^2/\lambda$~\cite{nearfield}, where $D_T$ and $D_R$ denote the aperture of the transmit and receive antenna array, respectively. In this way, the near-field effect is manifested so that $DoF_{\mathbf{H}}>1$.

	\subsection{Characteristic Mode Analysis of RHS Antenna}

	\begin{figure}[t]
		\centerline{\includegraphics[width=8cm,height=6cm]{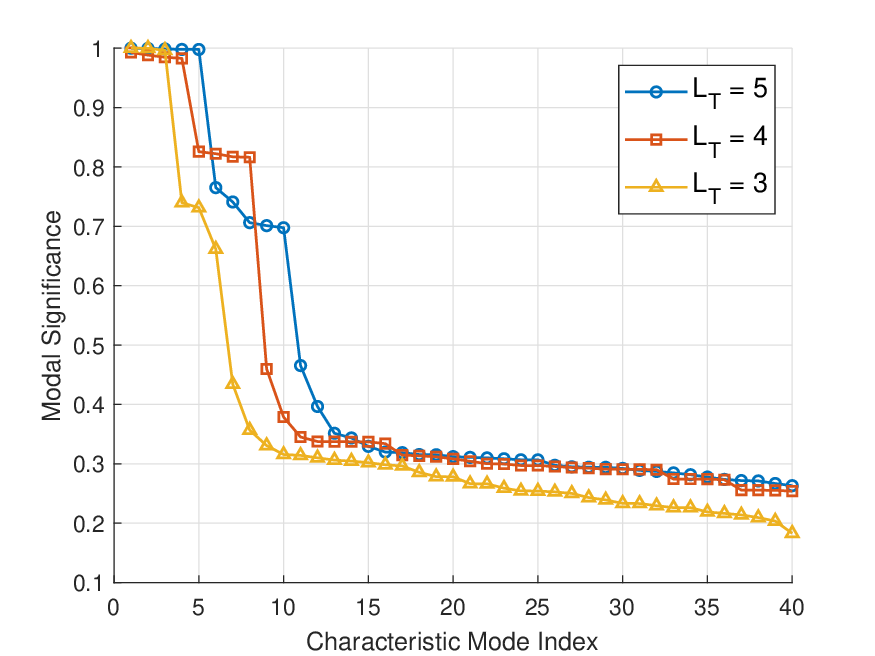}}
		\caption{Modal significance of the RHS antenna with different numbers of ports $L_T$.}
		\vspace{3mm}
		\label{fig::ms}
	\end{figure}
	Fig.~\ref{fig::jcma} demonstrates the current distribution of five characteristic modes with the strongest radiation ability for the RHS antenna when the port number is $4$.
	As shown in Fig.~\ref{fig::jcma}, the current distributions of different characteristic modes are nearly orthogonal in magnitude and polarization to each other, which verifies (\ref{equ::orthogonal}) and lays the foundation for the derivation of the achievable DoF of the CMA-modeled MIMO system.    
	
	Fig.~\ref{fig::ms} demonstrates the modal significance of characteristic modes for RHS with different numbers of ports. For the $i$-th characteristic mode, the modal significance is defined as $|\frac{1}{1+j\lambda_i}|$ and characterizes the radiation ability of the characteristic mode. That is to say, with the growth of the modal significance of a characteristic mode, the proportion of the mode's corresponding electric field in the overall electric field of the antenna increases. As shown in Fig.~\ref{fig::ms}, for RHS with different port numbers, the modal significance for the first few characteristic modes is close to $1$, i.e., the upper bound of modal significance, then it decreases rapidly. As the characteristic mode index continues to increase, the modal significance begins to decrease gradually, converging to $0$ given a significantly large number of characteristic modes. This trend shows that the number of characteristic modes for effective radiation is limited. Besides, with an increase in port numbers, the number of characteristic modes with considerable modal significance grows, which indicates that an enlarged antenna aperture size can contribute to the diversity of characteristic modes with strong radiation ability. The number of characteristic modes in the following simulation is set as 20 to include all characteristic modes for effective radiation.   
	
	\subsection{Evaluation of CMA-GA}
	\begin{figure}[t]
		\centerline{\includegraphics[width=8cm,height=6cm]{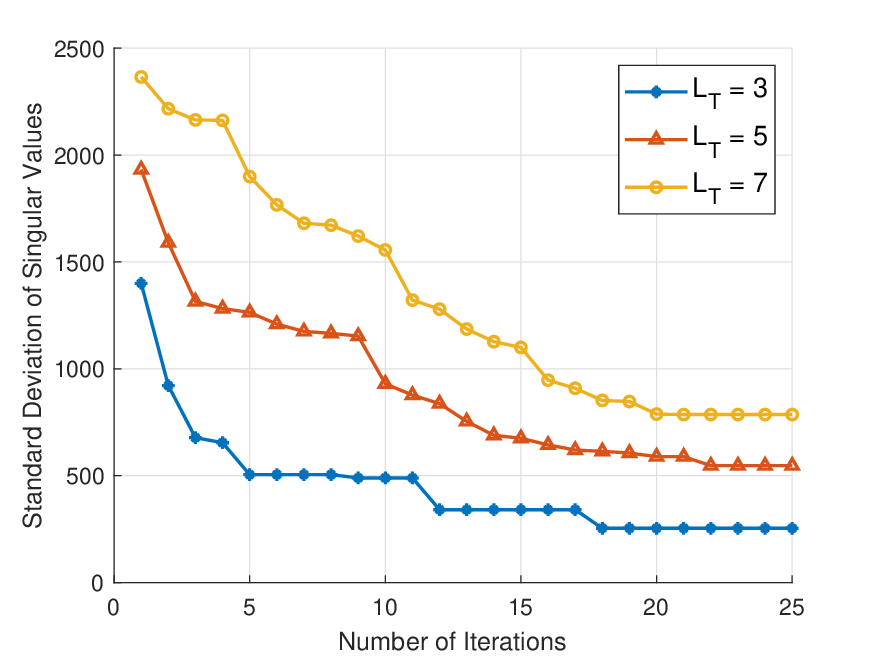}}
		\caption{The standard deviation of singular values concerning different numbers of iterations for the proposed CMA-GA.}
		\label{fig::converge}
	\end{figure}
	
Fig.\ref{fig::converge} presents the convergence of the proposed CMA-GA for different port numbers of the RHS antenna. As shown in Fig.\ref{fig::converge}, with the increase in iteration numbers, the standard deviation of the equivalent channel's singular values gradually decreases and then remains constant. Since the standard deviation is the negative value of the fitness function given in (\ref{equ::fit}), such a trend of curves demonstrates the convergence of CMA-GA. In this way, the difference of singular values is effectively reduced. Besides, a growth in the standard deviation is observed as the port numbers of RHS increase. This is because the difference in modal excitation coefficient between different ports become more obvious with the rise in port numbers. 

	\begin{figure}[t]
	\centerline{\includegraphics[width=8cm,height=6cm]{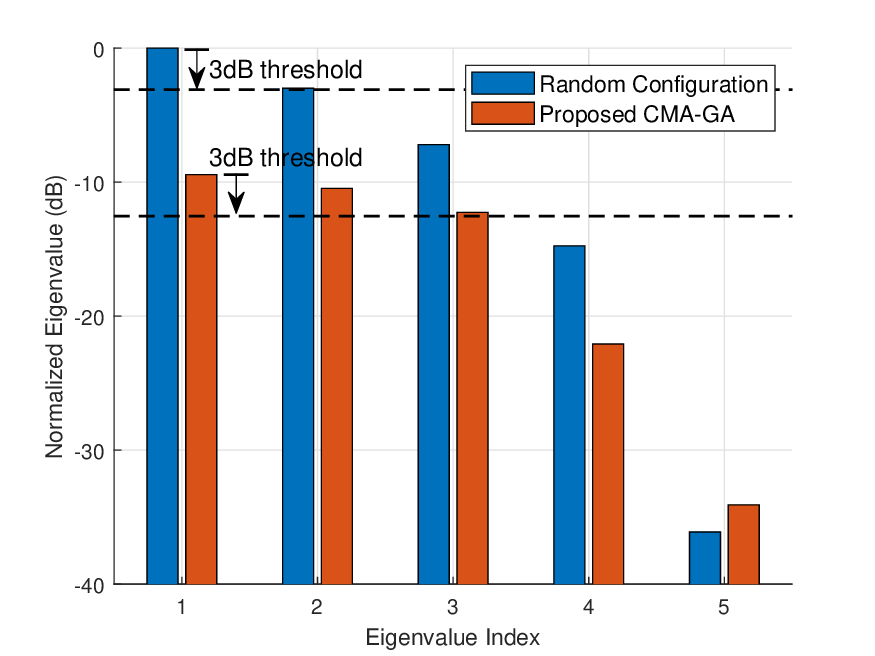}}
	\caption{Eigenvalues of a random and the optimized transceiver RHS configuration.}
	\vspace{3mm}
	\label{fig::optsvd}
\end{figure}

Fig.~\ref{fig::optsvd} demonstrates the eigenvalues distribution of the equivalent channel $\tilde{\mathbf{H}}$ with or without the optimization of CMA-GA. The port number is set as $L_T = L_R = 4$. As shown in Fig.~\ref{fig::optsvd}, for a random RHS configuration without the optimization via the proposed CMA-GA, eigenvalues decay rapidly with the increment of the eigenvalue index. On the contrary, for the RHS configuration optimized with CMA-GA, the eigenvalue distribution is more balanced. The achievable DoF $DoF_\mathbf{H}$ is calculated based on the eigenvalue distribution via (\ref{equ::numdof}), where the threshold $\gamma$ is set as $0.5$ \cite{multidof}, i.e., eigenvalue that is larger than the half of the largest eigenvalue is counted in $DoF_\mathbf{H}$. The threshold is illustrated in Fig.~\ref{fig::optsvd} with dash lines denoted with the 3dB threshold. It is observed that $DoF_\mathbf{H} = 3$ is obtained with the optimization of CMA-GA, which is larger than  $DoF_\mathbf{H} = 2$ with a random RHS configuration. Therefore, CMA-GA leads to a balanced eigenvalue distribution and an increase in $DoF_\mathbf{H}$.

\begin{figure}[t]
\centerline{\includegraphics[width=8cm,height=6cm]{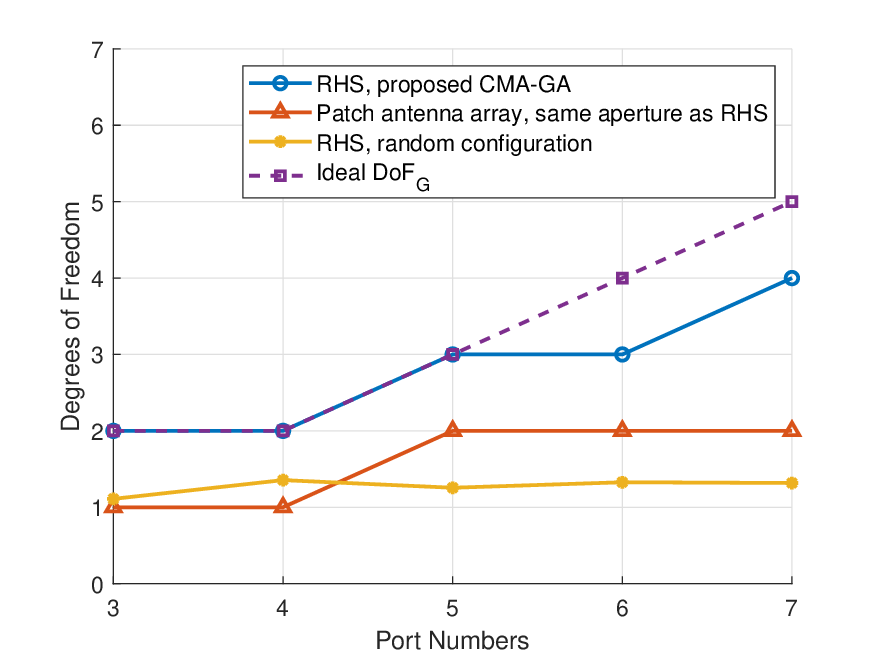}}
\caption{The DoF of the MIMO system of different port numbers.}
\vspace{3mm}
\label{fig::dofd}
\end{figure}

Fig.~\ref{fig::dofd} demonstrates the DoF of the MIMO system with different port numbers. The average $DoF_\mathbf{H}$ of the random RHS configuration and $DoF_\mathbf{H}$ of the patch antenna array with the same aperture as the RHS are presented as benchmarks. The ideal DoF, i.e., $DoF_{\mathbf{G}}$, is obtained where only the wireless propagation environment is considered and the influence of antennas is removed. 
	
As shown in Fig.~\ref{fig::dofd}, the ideal DoF is larger than the achievable DoF, i.e., $DoF_{\mathbf{G}} \geq DoF_{\mathbf{H}}$, which is in accordance with (\ref{equ::dofleq}) and verifies that the achievable DoF of a MIMO system is smaller than the DoF provided by the wireless channel. Besides, compared with the random RHS configuration and the case of patch antenna array with fixed EM properties, the configuration generated by the proposed CMA-GA reaches a higher $DoF_{\mathbf{H}}$, which manifests the ability of CMA-GA to improve the achievable DoF by optimizing the EM properties of RHS. Moreover, with the rise in port numbers, $DoF_{\mathbf{H}}$ obtained via CMA-GA and increases, but the average $DoF_{\mathbf{H}}$ obtained via random configurations fluctuates around $1.3$. This difference implies that an optimized RHS configuration is required to exploit the multiplexing gain provided by the enlarged aperture size.

	\section{Conclusion}\label{sec::conclusion}
	In this paper, we investigated the DoF of the MIMO system where the EM characteristics of antennas were taken into account. The CMA theory was exploited to model the excitation and radiation properties of the antennas via characteristic modes, based on which we derived the signal model and the achievable DoF of the CMA-modeled MIMO system. We formulated an achievable DoF maximization problem, and a case study where the RHS antennas were deployed at the transceiver of the MIMO was demonstrated. We proposed a CMA-GA algorithm to optimize the configuration of the RHS antennas so that the achievable DoF was maximized.
	
	Simulation results showed that: 1) The achievable DoF of the whole system is no larger than the DoF of wireless channel only because of the effect of transceiver antennas, verifying the theoretical analysis. 2) While increasing the number of antenna ports leads to more characteristic modes with large modal significance, it does not necessarily lead to a larger achievable DoF. In this case, optimizing RHS configurations can enhance the achievable DoF by changing the surface current and electric field distribution of characteristic modes. 3) With the RHS configuration obtained based on the proposed CMA-GA algorithm, the achievable DoF of the CMA-modeled MIMO system is increased, serving as a guideline for RHS configuration design.
	
	
\begin{appendices}
\section{Proof of Lemma \ref{prop::rx}}\label{app:rx}
The residual electric field $\bm{e}_{R,orth}$ is orthogonal to the space $\mathbb{E}_R$ formed by the electric field of all characteristic modes, i.e., $\{\bm{e}_{R,1}, \bm{e}_{R,2},...,\bm{e}_{R,n_R}\}.$

We denote $\overline{\mathbf{E}}_R = [\bm{e}_{R,1},\bm{e}_{R,2},...,\bm{e}_{R,n_R}]$, which is the matrix collecting all the electric fields of the characteristic modes. The projection of the received electric field on the space is given using the pseudo-inverse as
\begin{equation}
	\bm{e}_{R,proj} = \overline{\mathbf{E}}_R(\overline{\mathbf{E}}_R^H\overline{\mathbf{E}}_R)^{-1}\overline{\mathbf{E}}_R^H \bm{e}_{R}.
\end{equation}

The received electric field can be written as
\begin{equation}
	\sum_{n=1}^{n_R}{\beta_{n}m_{R,n}}\bm{e}_{R,n} = \overline{\mathbf{E}}_R {\rm diag}\{\bm{m}_R\}\bm{\beta}_R=\bm{e}_{R,proj},\label{equ::rxequ}
\end{equation}
where $\bm{m}_R$ and $\bm{\beta}_R$ are the vectors collecting all significance coefficients of characteristic modes of the receive antenna and the superposition coefficients, respectively. (\ref{equ::rxequ}) can be considered as a linear equation, and $\bm{\beta}_R$ is a variable to be obtained.

Since the number of characteristic modes is much smaller than the number of divided elements on the antenna, the pseudo-inverse is used to solve the equation $\overline{\mathbf{E}}_R {\rm diag}\{\bm{m}_R\}\bm{\beta}=\bm{e}_{R,proj}$. The vector superposition coefficient can be given as
\begin{equation}
	\bm{\beta}_R = {\rm diag}\{\bm{m}_R\}^{-1}(\overline{\mathbf{E}}_R^H\overline{\mathbf{E}}_R)^{-1}\overline{\mathbf{E}}_R^H \bm{e}_{R}.
\end{equation}

\section{Proof of Proposition\ref{prop::dofsimp}}\label{app::prop4}

First, the rank of channel correlation matrix equals the rank of channel matrix\cite{proofrank}, i.e.,
\begin{align}
	&DoF_{H} = {\rm Rank}(\tilde{\mathbf{H}}\tilde{\mathbf{H}}^H) = {\rm Rank}(\tilde{\mathbf{H}})=\nonumber\\
	&{\rm Rank}(\mathbf{V}_R^{\dagger} {\rm diag}\{\bm{m}_R\}^{-1}\overline{\mathbf{E}}_R^\dagger\mathbf{G}\overline{\mathbf{J}}_T {\rm diag}\{\bm{m}_T\}\mathbf{V}_T).\label{equ::rankhsimp}
\end{align}

Next, a decomposition is applied to the discrete wireless channel $\mathbf{G}$ based on the electric field and current of the characteristic modes, i.e., $\overline{\mathbf{E}}_R$ and $\overline{\mathbf{J}}_T$, which is denoted as
\begin{equation}
	\mathbf{G} = \overline{\mathbf{E}}_R\bm{\Gamma}\overline{\mathbf{J}}_T^T.\label{equ::decompG2}
\end{equation}

Considering the orthogonality of the normalized surface current between different characteristic modes of the transmit antenna, we have
\begin{equation}
	{\overline{\mathbf{J}}_T}^T\overline{\mathbf{J}}_T\approx \mathbf{I}_{n_T},\label{equ::orthogonal}
\end{equation}
where $\mathbf{I}_{n_T} \in \mathbb{Z}^{n_T \times n_T}$ is the identity matrix.
By substituting the decomposition based on characteristic modes, i.e., (\ref{equ::decompG2}), and the orthogonality of surface current, i.e., (\ref{equ::orthogonal}), into (\ref{equ::rankhsimp}), $\tilde{\mathbf{H}}$ is rewritten as

\begin{equation}
	\tilde{\mathbf{H}}\approx\mathbf{V}_R^{\dagger} {\rm diag}\{\bm{m}_R\}^{-1}\bm{\Gamma} {\rm diag}\{\bm{m}_T\}\mathbf{V}_T.
\end{equation}

Therefore, the achievable DoF of the MIMO system is given as
\begin{equation}
	DoF_{\mathbf{H}} ={\rm Rank}(\mathbf{V}_R^{\dagger} {\rm diag}\{\bm{m}_R\}^{-1}\bm{\Gamma} {\rm diag}\{\bm{m}_T\}\mathbf{V}_T),
\end{equation}
which completes the proof.

\section{Proof of Proposition \ref{prop::dofrelation}}\label{app::prop1}

First, the proof of (\ref{equ::dofleq}) is given. The equivalent channel of the CMA-modeled MIMO system is given as $\tilde{\mathbf{H}} = \mathbf{U}_R\mathbf{G}\mathbf{U}_T$. Hence, based on the property of matrix rank, we have
\begin{align}
	&DoF_{\mathbf{H}} \leq \min\{{\rm Rank}(\mathbf{V}_R^{\dagger} {\rm diag}\{\bm{m}_R\}^{-1}),{\rm Rank}(\mathbf{G}),\nonumber\\
	&{\rm Rank}({\rm diag}\{\bm{m}_T\}\mathbf{V}_T)\} = \min\{{\rm Rank}(\mathbf{V}_R^{\dagger} diag\{\bm{m}_R\}^{-1}),\nonumber\\
	&DoF_\mathbf{G},{\rm Rank}({\rm diag}\{\bm{m}_T\}\mathbf{V}_T)\}.
\end{align}

Hence, we have $DoF_{\mathbf{H}} \leq DoF_\mathbf{G}$.

Second, the proof of (\ref{equ::dofgeq}) is given. By referring to the definition of $DoF_{\mathbf{H}}$ in (\ref{equ::dofdef}) and Frobenius Rank Equality, we have
\begin{equation}
	DoF_{\mathbf{H}} = {\rm Rank}(\mathbf{V}_R^{\dagger} {\rm diag}\{\bm{m}_R\}^{-1}\bm{\Gamma} {\rm diag}\{\bm{m}_T\}\mathbf{V}_T)
\end{equation}

\begin{align}
	DoF_{\mathbf{H}} \geq& {\rm Rank}(\mathbf{V}_R^{\dagger} {\rm diag}\{\bm{m}_R\}^{-1}\bm\Gamma)+\nonumber\\
	&{\rm Rank}(\bm\Gamma {\rm diag}\{\bm{m}_T\}\mathbf{V}_T)-{\rm Rank}(\bm\Gamma).\label{equ::frobrank}
\end{align}
With the Sylvester Inequality, the first two terms on the right side of (\ref{equ::frobrank}) are lower-bounded by
\begin{align}
	&{\rm Rank}(\mathbf{V}_R^{\dagger} {\rm diag}\{\bm{m}_R\}^{-1}\bm\Gamma) \geq \nonumber\\
	&{\rm Rank}(\mathbf{V}_R^{\dagger} {\rm diag}\{\bm{m}_R\}^{-1}) + {\rm Rank}(\bm\Gamma) - n_R,\\
	&{\rm Rank}(\bm\Gamma {\rm diag}\{\bm{m}_T\}\mathbf{V}_T) \geq \nonumber \\
	&{\rm Rank}({\rm diag}\{\bm{m}_T\}\mathbf{V}_T) + {\rm Rank}(\bm\Gamma) - n_T,
\end{align}

By substituting the above two inequalities into (\ref{equ::frobrank}), we have 
\begin{align}
	DoF_{\mathbf{H}} \geq &{\rm Rank}(\mathbf{V}_R^{\dagger} {\rm diag}\{\bm{m}_R\}^{-1})+{\rm Rank}({\rm diag}\{\bm{m}_T\}\mathbf{V}_T)\nonumber \\
	+&{\rm Rank}(\bm\Gamma)-n_R-n_T.\label{equ::doflb2}
\end{align}
Because the diagonal matrices ${\rm diag}\{\bm{m}_R\}^{-1}$ and ${\rm diag}\{\bm{m}_T\}$ are invertible, we have ${\rm Rank}(\mathbf{V}_R^{\dagger}{\rm diag}\{\bm{m}_R\}^{-1})={\rm Rank}(\mathbf{V}_R^{\dagger})$ and ${\rm Rank}({\rm diag}\{\bm{m}_T\}\mathbf{V}_T)={\rm Rank}(\mathbf{V}_T)$. Therefore, (\ref{equ::doflb2}) is simplified as
\begin{equation}
	DoF_{\mathbf{H}} \geq {\rm Rank}(\mathbf{V}_R^{\dagger})+{\rm Rank}(\mathbf{V}_T) + {\rm Rank}(\bm\Gamma)-n_R-n_T.
\end{equation}

Finally, assuming the SVD of $\mathbf{V}_R$ is given as $\mathbf{V}_R = \mathbf{U}_{\mathbf{V}_R}\Sigma_{\mathbf{V}_R }\mathbf{V}_{\mathbf{V}_R}^{H}$, the SVD of $\mathbf{V}_R^{\dagger}$ can be denoted as $\mathbf{V}_R^{\dagger} = \mathbf{V}_{\mathbf{V}_R}\Sigma_{\mathbf{V}_R}^{-1}\mathbf{V}_{\mathbf{U}_R}^{H}$, where  $\Sigma_{\mathbf{V}_R}$ and $\Sigma_{\mathbf{V}_R}^{-1}$ satisfies
\begin{align}
	&[\Sigma_{\mathbf{V}_R}]_{i,j} = \frac{1}{[\Sigma_{\mathbf{V}_R}^{-1}]_{i,j}}, i=j,\\
	&[\Sigma_{\mathbf{V}_R}]_{i,j} = [\Sigma_{\mathbf{V}_R}^{-1}]_{i,j} =0, i \neq j,
\end{align} 

Therefore, we have 
\begin{equation}
	{\rm Rank}(\mathbf{V}_R^{\dagger}) = {\rm Rank}(\mathbf{V}_R),
\end{equation}
which completes the proof.
\end{appendices}

\end{document}